\documentclass[]{interact}

\usepackage{amsthm,amsmath}
\usepackage{amssymb}
\usepackage{pstricks,pst-node,pst-text,pst-3d}
\usepackage[colorlinks,citecolor=blue]{hyperref}

\usepackage{epsfig,graphicx}
\usepackage{epstopdf}% To incorporate .eps illustrations using PDFLaTeX, etc.
\usepackage[caption=false]{subfig}% Support for small, `sub' figures and tables

\usepackage[numbers,sort&compress]{natbib}% Citation support using natbib.sty
\bibpunct[, ]{[}{]}{,}{n}{,}{,}% Citation support using natbib.sty
\renewcommand\bibfont{\fontsize{10}{12}\selectfont}% Bibliography support using natbib.sty

\theoremstyle{plain}% Theorem-like structures provided by amsthm.sty
\newtheorem{theorem}{Theorem}[section]
\newtheorem{lemma}[theorem]{Lemma}

\theoremstyle{definition}

\theoremstyle{remark}

\newcommand{\be}{\begin{equation}}
\newcommand{\ee}{\end{equation}}

\newcommand{\bx}{\mbox{\bf x}}

\newcommand{\bgamma}{\mbox{\boldmath $\gamma$}}

\begin{document}

%\articletype{Original Research Paper}% Specify the article type or omit as appropriate

\title{Testing Homogeneity in a heteroscedastic contaminated normal mixture}

\author{
\name{Xiaoqing Niu\textsuperscript{1}
and Pengfei Li\textsuperscript{2}
and Yuejiao Fu\textsuperscript{3}\thanks{CONTACT Yuejiao Fu. Email: yuejiao@mathstat.yorku.ca}}
\affil{\textsuperscript{1}Toronto Dominion Bank Group, Toronto, ON Canada; 
\textsuperscript{2}Department of Statistics and Actuarial Science, University of Waterloo, Waterloo, ON Canada;
\textsuperscript{3}Department of Mathematics and Statistics, York University, Toronto, ON Canada}
}

\maketitle

\begin{abstract}
Large-scale simultaneous hypothesis testing appears in many areas such as
microarray studies, genome-wide association studies, brain imaging, disease mapping and astronomical surveys.
A well-known inference method is to control the false discovery rate.
One popular approach is to model the $z$-scores derived from the individual $t$-tests
and then use this model to control the false discovery rate.
We propose a new class of contaminated normal mixtures for modelling $z$-scores.
We further design an  EM-test for testing homogeneity in this class of mixture models.
We show that the EM-test statistic has a shifted mixture of chi-squared
limiting distribution.
Simulation results show that the proposed testing procedure has accurate type I error
and significantly larger power than its competitors under a variety of model specifications.
A real-data example is analyzed to exemplify the application of the proposed method.
\end{abstract}

\begin{keywords}
Computer experiment;
heteroscedastic contaminated normal model;
homogeneity test;
likelihood ratio test; EM-test.
\end{keywords}

\section{Introduction}

Large-scale simultaneous hypothesis testing appears in many applications such as
microarray studies, genome-wide association studies, brain imaging, disease mapping and astronomical surveys. This paper is motivated by  large-scale multiple testing problem, in which thousands and sometimes millions of
hypothesis tests need to be conducted on parallel data sets.
The prostate data discussed in \cite{efron2010large} is an excellent illustrating example for the large-scale multiple testing problem.
The data consists of the gene expression levels of $n=6033$ genes for 102 male individuals:
52 prostate cancer patients and 50 normal control subjects.
The goal  of the study is to find the genes that
are differentially expressed between the two groups of subjects. For the $i$th gene, we
test
$$
H_{0i}: \mbox{gene $i$ is not differentially expressed in the two samples or gene $i$ is ``null"}.
$$
Hence over 6000 hypotheses need to be tested at the same time.
The commonly used test statistic for $H_{0i}$ is the traditional two-sample $t$-test statistic $t_i$.
Under $H_{0i}$, $t_{i}$ follows or approximately follows a $t$-distribution with 100 degrees of freedom.
\citet{efron2010large} suggested transforming $t_i$ to $z_i$
by
$
z_i=\Phi^{-1}(F_{100}(t_i)),
$
where $\Phi$ and $F_{100}$ are the cumulative distribution functions for the standard normal distribution $N(0,1)$
and the $t$-distribution with 100 degrees of freedom.
Under $H_{0i}$, the $z$-score $z_i$ follows or approximately follows  $N(0,1)$, known as the theoretical null distribution.
%Efron (2010) then used $z_i$ to test $H_{0i}$.
%Let $y_{ij}$ be the expression level for gene $i$ and on individual $j$,
%$i=1,\ldots, 6033$ and $j=1,\ldots, 102$;
%with $j=1,\ldots,50$ for the normal controls and $j=51,\ldots, 102$ for the cancer patients.

In large-scale multiple testing problems, scientists are not interested in controlling
the type-I error individually. Instead they prefer the notion of  controlling the false discovery rate
\cite{benjamini1995controlling}.
Among many methods for controlling this rate,
\citet{efron2004large} proposed  the use of a finite normal mixture
to model the $z$-scores.
See also \cite{mclachlan2006simple} and \cite{dai2010contaminated}.
An appropriate candidate model is
$$
(1-\alpha) f(x; \mu_1, \sigma_1)+\alpha f(x;\mu_2, \sigma_2)
$$
with the first component corresponding to the null genes
and the second component corresponding to the differentially expressed genes.
Here $ f(x;\mu, \sigma)$ denotes the probability density function (pdf)
of the normal distribution $N(\mu,\sigma^2)$.
In theory, $f(x; \mu_1, \sigma_1)$ should be the pdf of $N(0,1)$, the theoretical null distribution.
In practice, the theoretical null distribution may fail to work due to several reasons.
For example,
the null distribution of $t_i$ may not be the exact $t$-distribution, see Efron (2010, pp.\ 105--109).
To solve the problem, \citet{efron2010large} suggested the notion of empirical null distribution,
which can be estimated from the data itself and in many examples can be well approximated by $N(0,\sigma_1^2)$.
The use of an empirical null distribution rather than the theoretical null distribution can be critical for making correct inference.
Along this line, \citet{dai2010contaminated} suggested a homoscedastic  contaminated normal mixture model
\begin{equation}
\label{chap4.model2}
(1-\alpha)f(x; 0,\sigma)+\alpha f(x; \mu,\sigma).
\end{equation}
Before  identifying  the genes that are differentially expressed in the two samples,
one should first detect the existence  of differentially expressed genes.
To this end, \citet{dai2010contaminated} proposed two methods, the MLRT and the D-test, to test homogeneity in (\ref{chap4.model2}).

Homogeneity in variance assumption is often violated in practice (see the real example in Section \ref{chap4 con eg}). To relax
the common variance assumption, in this paper,  we suggest modelling the $z$-scores by
a heteroscedastic contaminated normal mixture model
\begin{equation}
\label{chap4.model}
(1-\alpha) f(x; 0, \sigma_1)+\alpha f(x;\mu, \sigma_2).
\end{equation}
We wish to test
\begin{equation}
\label{chap4.test}
H_0:\alpha=0 \textrm{ or } (0, \sigma_1)=(\mu, \sigma_2).
\end{equation}
Developing an effective testing procedure for (\ref{chap4.test})
is a challenging problem.
The log-likelihood function is unbounded \cite{chen2008inference}; and the Fisher information in the mixing proportion direction may be infinity \cite{chen2009hypothesis}. The asymptotic results for such existing methods as the likelihood ratio test can not be directly applied \cite{dacunha1999testing,liu2003asymptotics}. To tackle the problem, we adopt the idea of the EM-test which was originally developed by  \citet{li2009non} and \citet{chen2009hypothesis}. Utilizing the special feature of the heteroscedastic contaminated normal mixture model (\ref{chap4.model}), we design a new EM-test specifically for hypothesis (\ref{chap4.test})
and develop its asymptotic properties.

For comparative purposes, we consider another stream for detecting the existence of differentially expressed genes
which is based on $p$-values: $p_i=2\{1-\Phi(|z_i|)\}$, $i=1,\ldots,n$.
A commonly used model for modelling $p$-values is
the following contaminated Beta model
\begin{equation}
\label{chap4.model3}
(1-\alpha)B(1,1)+\alpha B(a,b),
\end{equation}
where $B(a,b)$ denotes the Beta distribution with two parameters $a$ and $b$.
Note that $B(1,1)$ is also  the uniform distribution over (0,1).
See \cite{allison2002mixture,peng2003simultaneous,dai2008omnibus},
and the reference therein.
\citet{dai2008omnibus} further proposed the use of MLRT and D-test for
testing homogeneity in (\ref{chap4.model3}). One drawback of methods based on the contaminated Beta model (\ref{chap4.model3}) is that when some  doubt has been casted on the adequacy of the theoretical null distribution, the simulated type-I error of the tests will be severely inflated as shown in our simulation studies in Section \ref{chap4 con simu}. In sharp contrast to the methods based on contaminated Beta model, the tests based on model (\ref{chap4.model2}) and the proposed EM-test based on model (\ref{chap4.model}), are invariant to the scale transformation and can be used efficiently when an empirical null distribution should be used for the large-scale hypothesis testing. Our simulation studies in Section \ref{chap4 con simu} shows that under various circumstance, the proposed EM-test has accurate levels and superior power compared with its competitors.

The rest of this paper is organized as follows.
In Section \ref{chap4 con em}, we provide a complete description of the new EM-test procedure,
and its asymptotic properties.
Simulation studies are conducted in Section \ref{chap4 con simu} to evaluate the empirical performance
of the proposed method and some alternatives.
A real-data example
is given in Section \ref{chap4 con eg}. Proofs are given in the Appendix.

\section{Main Results \label{chap4 con em}}
We propose an EM-test for the homogeneity problem (\ref{chap4.test}). We first give a detailed description of our testing procedure and then establish its theoretical foundations and asymptotic properties.

\subsection{EM-test}
\label{chap4 con proc}
Suppose $X_1, \ldots, X_n$ are a random sample
of size $n$ from the contaminated normal mixture model (\ref{chap4.model}). We denote the log-likelihood function as
\[
l_n(\alpha, \mu, \sigma_1, \sigma_2)
=\sum_{i=1}^n \log \{ (1-\alpha)f(X_i; 0, \sigma_1)+\alpha f(X_i; \mu, \sigma_2) \}
\]
and define the modified log-likelihood function as
\[
{pl_n}(\alpha, \mu, \sigma_1, \sigma_2)
=l_n(\alpha, \mu, \sigma_1, \sigma_2)+ p(\alpha)
+p_n(\sigma_1)+p_n(\sigma_2).
\]
The penalty function $p(\alpha)$
is used to prevent the fitting of $\alpha$ being close to 0.
Note that  in large-scale hypothesis testing problem, the true value of $\alpha$
is in general quite small, for example, smaller than 0.25 (Efron, 2010). Therefore, we do not need to penalize the fitting of $\alpha$ being close to 1 .
One choice for $p(\alpha)$ could be $p(\alpha)= \log (\alpha) $, which has been used in Fu, Chen, and Li (2008)
for testing homogeneity in a class of contaminated von Mises model.
The penalty $p_n(\sigma)$ prevents the fitting of $\sigma_1^2$ and $\sigma_2^2$
being close to 0, which help avoid  the unbounded likelihood \citep{chen2008inference}.
An example of $p_n(\sigma)$ is
$$
p_n(\sigma)=-a_n \cdot \left( \frac{\hat \sigma_0^2}{\sigma^2} + \log \frac{\sigma^2}{\hat \sigma_0^2} \right),
$$
where $\hat \sigma_0^2=\sum_{i=1}^nX_i^2/n$ is the maximum likelihood estimator
of the variance parameter under the null model.
Using the penalty function, $\hat \sigma_0^2$ also maximizes the modified log-likelihood function under the null hypothesis.
The choice of $a_n$ is  discussed in Section \ref{chap4 con tune}.

Building on the modified log-likelihood function, we develop an EM-test for the heteroscedastic contaminated normal mixture model.
We first choose  a finite set
$\{\alpha_1 , \ldots, \alpha_J\}$ as the initial values for
$\alpha$ and a positive integer $K$ as the number of iterations.
As mentioned earlier, in most large-scale testing problems, the value of $\alpha$ is quite small.
We recommend using $\{0.05, 0.15, 0.25\}$ in practice. Empirical experience suggests that
further increasing the number of initial values will not significantly improve the power of the EM-test.
We suggest use $K=3$ as the number of iterations. The EM-test statistics are constructed in the following steps.

Step 1. Set $k=1$. For $j=1,2, \ldots, J$, set $\alpha_j^{(1)}=\alpha_j$ and compute
\[
(\mu_j^{(1)}, \sigma_{j,1}^{(1)}, \sigma_{j,2}^{(1)})
=\arg\max_{\mu, \sigma_1, \sigma_2}
pl_n(\alpha_j, \mu, \sigma_1, \sigma_2).
\]

Step 2. For the current $k$, use an E-step to
compute the posterior probabilities
\[
w_{ij}^{(k)}=\frac{\alpha_j^{(k)}f(X_i; \mu_j^{(k)}, \sigma_{j, 2}^{(k)})}
{(1-\alpha_j^{(k)})f(X_i; 0, \sigma_{j, 1}^{(k)})
+\alpha_j^{(k)}f(X_i; \mu_j^{(k)}, \sigma_{j, 2}^{(k)})}.
\]
Update $(\alpha,\mu,\sigma_1,\sigma_2)$ via an M-step such that
\begin{eqnarray*}
\alpha_j^{(k+1)} &=&
\arg \max_{\alpha } \Big\{
 \sum_{i=1}^n(1- w_{ij}^{(k)})\log(1-\alpha) +\sum_{i=1}^n w_{ij}^{(k)}\log(\alpha)+p(\alpha)
\Big\},
\\
\mu_j^{(k+1)} &=&
\arg \max_{\mu}
\Big \{
\sum_{i=1}^n w_{ij}^{(k)}\log f(X_i;\mu, \sigma_2)
\Big \},
\\
\left(\sigma_{j, 1}^{(k+1)},\sigma_{j, 2}^{(k+1)}\right) &=&
 \arg \max_{\sigma_1,\sigma_2}
 \Big\{
 \sum_{i=1}^n(1- w_{ij}^{(k)})\log f(X_i; 0, \sigma_1)
 +p_n(\sigma_1)+p_n(\sigma_2)\\
 &+&
 \sum_{i=1}^n w_{ij}^{(k)}\log f(X_i;\mu_j^{(k)}, \sigma_2)
  \Big\}.
\end{eqnarray*}
Iterate the E-step and M-step $K-1$ times.

Step 3. For each $k$ and $j$, define
\[
M_{n}^{(k)} (\alpha_j)=
2\{pl_n (\alpha_j^{(k)}, \mu_j^{(k)}, \sigma_{j,1}^{(k)},\sigma_{j,2}^{(k)})
- pl_n (1, 0, \hat {\sigma}_0, \hat {\sigma}_0)\}.
\]
The EM-test statistic $EM_n^{(K)}$ is then
$$
EM_n^{(K)}=\max\{M_n^{(K)}(\alpha_j), j=1,2, \ldots, J\}.
$$

Finally, the null hypothesis is rejected when $EM_n^{(K)}$
exceeds the critical value of the limiting distribution
given in the next section.

\subsection{Asymptotic properties}
\label{chap4 con prop}

In this section, we derive the limiting distribution of the EM-test statistics $EM_{n}^{(K)}$
under the following conditions on the penalty functions
$p(\alpha)$ and $p_n(\sigma)$.

\begin{enumerate}
\item[D1]
The penalty function $p(\alpha)$ is continuous,
approaches negative infinity as $\alpha$ approaches zero.
Further $p(1)=0$.

\item[D2]
$\sup \{ |p_n(\sigma)| \}=o(n)$.

\item[D3]
$p'_n(\sigma)=o_p (n^{1/6})$ at any $\sigma >0$.

\item[D4]
$p_n(\sigma)\leq 4 (\log n)^2 \log (\sigma)$,
when $\sigma \leq n^{-1}$
and $n$ is large.

\end{enumerate}

Conditions D1 to D4 are weak technical conditions, which guarantee that the EM-test has a simple limiting distribution. The aforementioned choice of penalty functions all satisfy  these conditions.

\begin{theorem}
\label{chap4 con thm2}
Suppose that the penalty functions $p(\alpha)$,
$p_n(\sigma)$
satisfy Conditions D1-D4 and
the initial set $\{\alpha_1, \ldots, \alpha_J\}\in (0,1)$.
Under the null hypothesis
and for any fixed finite $K$, as $n \to \infty$,
$$
EM_{n}^{(K)} \to \frac{1}{2}\chi_1^2 + \frac{1}{2}\chi_2^2+ 2 \max_{j} p(\alpha_j)
$$
in distribution.
\end{theorem}

In the following, we give some insight in understanding  Theorem \ref{chap4 con thm2}.
Testing homogeneity in (\ref{chap4.model})
is equivalent to testing $\mu=0$ and the homogeneity of variances of the two component distributions.
For any given $\alpha_j$, $M_n^{(K)}(\alpha_j)$ contains two parts: a term due to log-likelihood difference
and a term due to penalty functions.
Testing $\mu=0$ contributes a $\chi^2_1$ to the limiting distribution of the term due to
 log-likelihood difference;
testing the homogeneity of variances contributes to another $0.5\chi^2_0+0.5\chi^2_1$
to the limiting distribution of
the term  due to log-likelihood difference.
Roughly speaking, adding two parts together,
the term due to log-likelihood difference in $M_n^{(K)}(\alpha_j)$ has a   $\frac{1}{2}\chi_1^2 + \frac{1}{2}\chi_2^2$
 limiting distribution.
 Hence $EM_n^{(K)}$ has the limiting distribution given in Theorem \ref{chap4 con thm2}
 with  $2 \max_{j} p(\alpha_j)$ coming from the penalty functions.

\subsection{Choice of the penalty functions}
\label{chap4 con tune}

In this section, we address how to choose appropriate  penalty functions $p(\alpha)$ and
$p_n(\sigma)$. For $p(\alpha)$, we use the penalty $p(\alpha)=\log(\alpha)$
as suggested in \cite{fu2008modified}. This penalty function satisfies
Condition D1 for the theoretical development. Further, in the M-step of the EM-iteration,
$\alpha$ can be updated via an explicit form.
For the penalty $p_n(\sigma)$, we recommend
$$
p_n(\sigma)=-a_n \cdot \left( \frac{\hat \sigma_0^2}{\sigma^2} + \log \frac{\sigma^2}{\hat \sigma_0^2} \right).
$$
As long as $a_n=o_p(n^{1/6})$, $p_n(\sigma)$ satisfies Conditions D2-D4.
Carefully tuning the value of $a_n$
in the penalty functions
will further improve the precision of the approximation
of the limiting distribution to the finite-sample distribution of $EM_{n}^{(K)}$.
We adopt the computer experiment idea from \citet{chen2011tuning} to
obtain an empirical formula for $a_n$.

We first carry out pilot experiment for
many choices of  sample sizes $n$.
We find that when $a_n\leq 1.4$,  the simulated type-I errors of the EM-test are larger than
the nominal levels.
Hence, in the designed experiment, $a_n$ is chosen to be from 1.6 to 4.0,
with the step length 0.2. In total, 13 values of $a_n$ are considered.
We consider three quite large sample sizes: 500, 1000, and 1500,
since in large-scale hypothesis testing problem, the number of parallel hypotheses are around thousands.
Then a $13\times 3$ full factorial design is used in our computer experiment.
For each combination of $a_n$ and $n$,
5000 random samples of size $n$ from the $N(0,1)$ are used to calculate the simulated type-I errors $\hat q$
of $EM_n^{(3)}$ at the target significance level $q$.
The discrepancy between $\hat q$ and $q$
is calculated as
$$
y=\log\{\hat q/(1-\hat q)\}
-\log\{q/(1-q)\}.
$$
Table \ref{chap4 con pdiff} presents the discrepancy between $\hat q$ and $q$ when $q=0.05$.
{
\begin{table}[!hbp]
\caption{Discrepancy between $\hat q$ and $q$ in term of $y$  under the contaminated normal mixture models}
\vspace {0.1in}
\renewcommand{\arraystretch}{1.2}
\centering {
\hspace{0.0in}
\tabcolsep=1.2mm
\begin{tabular}{lrrrrrrrrrrrrr}
\hline

\hline
 $n \backslash a_n$& 1.6& 1.8& 2.0& 2.2& 2.4& 2.6& 2.8& 3.0& 3.2& 3.4& 3.6& 3.8& 4.0\\ \hline
 %500&0.18&  0.06&  0.10&  0.08&  0.08&  0.00& -0.02& -0.04& -0.11& -0.04& -0.11&  0.00& -0.21 \\
 500&0.175&  0.061&  0.101&  0.081&  0.081&  0.000& -0.021& -0.043& -0.111& -0.043& -0.111&  0.000& -0.208\\
 %1000&0.18&  0.10&  0.08&  0.12&  0.04&  0.10&  0.04&  0.08&  0.04&  0.00& -0.04&  0.04& -0.04 \\
 1000&0.175&  0.101&  0.081&  0.120&  0.041&  0.101&  0.041&  0.081&  0.041&  0.000& -0.043&  0.041& -0.043\\
  %1500&0.30&  0.21&  0.16&  0.10&  0.14&  0.10&  0.06& -0.02& -0.02&  0.14&  0.08&  0.06&  0.08 \\
   1500& 0.295&  0.210&  0.157&  0.101&  0.138&  0.101 & 0.061& -0.021& -0.021&  0.138 & 0.081 & 0.061&  0.081\\
\hline
\end{tabular} }
\label{chap4 con pdiff}
\end{table}
}

Analysis of variance suggests both $n$ and $a_n$
have significant effects on $y$.
After some brainstorming and exploratory analysis,
the covariates in the form of $1/n$ and $\log(a_n-1.4)$ gives the most
satisfactory outcomes in terms of both the goodness of fit and the simplicity of the resulting
formula for $a_n$.
The covariate $\log(a_n-1.4)$ effectively confines the value of $a_n$ in $(1.4,\infty)$,
as suggested by our pilot study.  We next  regress $y$ in $1/n$ and  $\log(a_n-1.4)$. Based on $39$ observations,
the fitted regression model is
\begin{equation*}
%\hat y=0.158-82.899/n-0.094\log(a_n-1.4)
\hat y=0.159-76.775/n-0.091\log(a_n-1.4)
\end{equation*}
with Adjusted $R^2=71\%$ .
Setting $\hat y=0$ gives the following empirical formula for $a_n$:
\begin{equation*}
 a_n=\exp(1.747-843.681/n)+1.4.
\end{equation*}
Since our EM-test procedure is invariant to the scale transformation,
the empirical formula of $a_n$ is applicable to the general null distribution $N(0,\sigma^2)$.
In the next section, we examine the performance of the derived empirical formula of $a_n$
and other suggested tuning parameters.

\section{Simulation Studies}
\label{chap4 con simu}

We have conducted simulation studies to evaluate the finite-sample performance of the proposed EM-test ($EM_n^{(K)}$), the MLRT ($\lambda_{n,N}$) and the D-test ($d_{n,N}$)  proposed in \cite{dai2010contaminated} under the homoscedastic contaminated normal mixture model (\ref{chap4.model2}),
and the MLRT ($\lambda_{n,B}$) and the D-test ($d_{n,B}$) proposed in \cite{dai2008omnibus}
under the contaminated Beta mixture model  (\ref{chap4.model3}).
Note that all the five methods could be used to
detect the existence of genes that are differentially expressed between two comparison groups.

We generated 10000 replications respectively for sample size
$n=$ 100, 500, 1000, 1500, 3000, and 10000
from the theoretical null distribution $N(0,1)$, and two different empirical null distributions $N(0, 0.9^2)$ and $N(0,1.1^2)$.
Since $EM_n^{(K)}$, $\lambda_{n,N}$, and $d_{n,N}$ are invariant to the scale transformation, we only need to consider the simulated type-I error rate under the theoretical null distribution. While for $\lambda_{n,B}$ and $d_{n,B}$, we considered both theoretical and empirical null distributions. The simulated type-I error rates of various tests at the nominal level  $5\%$
are reported in Table \ref{chap4 con nulltable}.

For $EM_n^{(K)}$,  the test statistics were calculated based on the recommended tuning parameters in Section \ref{chap4 con tune} and the limiting distribution in Theorem \ref{chap4 con thm2}
was used to calculate the critical values. As we can see that, for all considered sample sizes,
the simulated type-I errors are very close to the nominal levels, which indicates
that the limiting distribution approximates the finite sample distribution reasonably well. The simulated levels of $\lambda_{n,N}$ and $d_{n,N}$ are reasonable good as well.  In addition,  Table \ref{chap4 con nulltable} also reveals a fact that a slightly departure from theoretical null model can seriously affect the levels of the two tests $\lambda_{n,B}$ and $d_{n,B}$ which are constructed based on the contaminated Beta model. When the data are generated from the theoretical null model, $\lambda_{n,B}$ and $d_{n,B}$ have slightly inflated type-I error rates with small sample size. While when the sample size getting bigger, the simulated levels of $\lambda_{n,B}$ and $d_{n,B}$ become accurate.

{
\begin{table}[!htbp]
\caption{Simulated type-I error rates (\%) of various tests at the nominal level 5\%}
\vspace {0.1in}
\renewcommand{\arraystretch}{1.2}
\centering {
\hspace{0.0in}
\tabcolsep=1.5mm
\begin{tabular}{crrrrrrrrr}
\hline

\hline
%&&\multicolumn{3}{c}{Level=$10\%$}&\multicolumn{4}{c}{Level=$5\%$}&\multicolumn{3}{c}{Level=$1\%$}\\
%\cline{2-4} \cline{6-8}\cline{10-12}
\hline \hline
$n$      & $EM_n^{(3)}$& $\lambda_{n,N}$& $d_{n,N}$ &$\lambda_{n,B}$&$d_{n,B}$ &$\lambda_{n,B}$&$d_{n,B}$ &$\lambda_{n,B}$&$d_{n,B}$\\ \hline
Null & $N(0,1)$ & $N(0,1)$& $N(0,1)$&$N(0,1)$ &$N(0,1)$ & $N(0,0.9^2)$&$N(0,0.9^2)$& $N(0,1.1^2)$&$N(0,1.1^2)$\\ \hline
 100     &5.1 & 5.3&   5.6& 5.9& 12.8& 22.1& 23.4& 22.9& 35.5\\
 500     &4.9 & 4.9&   4.9& 5.1& 6.1 &83.2 &79.6& 78.4& 81.7\\
 1000   &4.7 & 5.1&   5.1& 5.0&  5.5 &99.4& 99.0& 97.5& 98.1\\
 3000   &4.9 & 4.9&.  4.9& 4.9& 5.1 &100.0& 100.0& 100.0& 100.0\\
 10000 & 5.1& 4.8&   4.8& 4.7& 4.6 &100.0 &100.0& 100.0& 100.0\\
\hline
\end{tabular} }
\label{chap4 con nulltable}
\end{table}
}
%Note: the a_n used for n listed in the null table are
%&  1.4&  1.45&  1.7&  2   &  2.3&  2.65&  2.95&  3.2&  3.4&  3.65
%&  4.4&  4.85&  5.2 &  5.4&  5.6&  5.75&  5.85&  5.95&  6.35

%{
%\begin{table}[!ht]
%\caption{Twelve alternative contaminated normal models.}
%\vspace {0.1in}
%\renewcommand{\arraystretch}{1.2}
%\centering {
%\hspace{0.0in}
%\tabcolsep=2mm
%\begin{tabular}{ccc}
%\hline
%No.&Model&\\
%\hline
%A1&$0.95 N(0,1)+0.05 N(1, 1)$&\\
%A2&$0.95 N(0,1)+0.05 N(2, 1)$&\\
%A3&$0.9 N(0,1)+0.1 N(1, 1)$&\\
%A4&$0.9 N(0,1)+0.1 N(2, 1)$&\\
%A5&$0.9 N(0,1)+0.05 N(1, 1)+0.05 N(-1, 1)$&\\
%A6&$0.9 N(0,1)+0.05 N(2, 1)+0.05 N(-2, 1)$&\\
%A7&$0.95 N(0,1)+0.05 N(1, 2)$&\\
%A8&$0.95 N(0,1)+0.05 N(1, 0.5)$&\\
%A9&$0.9 N(0,1)+0.1 N(2, 2)$&\\
%A10&$0.9 N(0,1)+0.1 N(2, 0.5)$&\\
%A11&$0.9 N(0,1)+0.05 N(1, 2)+0.05 N(-1, 0.5)$&\\
%A12&$0.9 N(0,1)+0.05 N(2, 2)+0.05 N(-2, 0.5)$&\\
%\hline
%\end{tabular} }
%\label{chap4 con altermodel}
%\end{table}
%}

For the simulated power, the results were based on 10000 repetitions. Based on the observation from Table \ref{chap4 con nulltable}, for both $\lambda_{n,B}$ and $d_{n,B}$, we used the critical values of the observed test statistics under the null to determine the simulated power in order to make the power comparison fair. To simulate data from alternative models, we considered various combinations of the mixing proportion $\alpha=$ 0.05, 0.07 and 0.1, the component variance $\sigma_2^2=$ 1 and 2,  and the mean $\mu=$ 1, 1.5 and 2.
In all these eighteen alternative models, the component variance $\sigma_1^2$ was set to be 1. The simulated power of the five testing procedures are listed in Table \ref{chap4 con altertable}. We only presented the power comparison
with sample size $n$=500  at the 5\% level in Table \ref{chap4 con altertable}.
The simulation results for other sample sizes and significance levels
show the similar trend and therefore are omitted. Based on Table \ref{chap4 con altertable}, we see that the proposed EM-test has superior performance than its competitors, especially for the situations with heteroscedastic component variances.

%From Table \ref{chap4 con altertable},
%we have the following two observations:
%\begin{enumerate}
%\item[(i)] $EM_n^{(K)}$ is more powerful than $\lambda_{n,N}$ and $d_{n,N}$
%under A2 and A9, in which we may observe quite large $z$-scores.
%$EM_n^{(K)}$ is also more powerful than $\lambda_{n,N}$ and $d_{n,N}$
%under A5-A6 and A11-A12, in which the estimator of $\mu$ is very close to 0.
%
%\item[(ii)] $EM_n^{(K)}$ is more powerful than or comparable to
%$\lambda_{n,B}$  and $d_{n,B}$ under A1-A4 and A7-A10.
%However, $EM_n^{(K)}$ is less power than $\lambda_{n,B}$  and $d_{n,B}$
%under A5-A6 and A11-A12.
%\end{enumerate}
%Based on these two observations, we conclude that
%(i)  the new EM-test based on model (\ref{chap4.model})
%is as powerful as or more powerful
%than the MLRT and D-test based on model (\ref{chap4.model2});
%(ii) the new EM-test based on model (\ref{chap4.model})
%is as powerful as or more powerful
%than the MLRT and D-test based on model (\ref{chap4.model3})
%when $z$-scores are generated from model (\ref{chap4.model});
%(iii) the new EM-test becomes less power than the MLRT and D-test based on model (\ref{chap4.model3}) when there is a symmetry between overexpression and underexpression.

{
\begin{table}[!htbp]
\caption{Power (\%) of various tests with $n=500$}
\vspace {0.1in}
\renewcommand{\arraystretch}{1.2}
\centering {
\hspace{0.0in}
\tabcolsep=2.5mm
\begin{tabular}{ccc ccccc }
\hline

\hline
   $\alpha$& $\sigma_2^2$ & $\mu$& $EM_n^{(3)}$& $\lambda_{n,N}$&$d_{n,N}$&$\lambda_{n,B}$&$d_{n,B}$\\ \hline %\hline
   0.05 &   1 &1.0&   18.9 &     19.4  & 19.3   &   10.0 &  10.6\\
   0.05 &   1 &1.5&   44.0 &     38.6 &  38.7   &  33.0  &  35.8\\
   0.05 &   1 &2.0&   78.2 &     64.5  & 65.0   &   75.7 &  79.1\\
   0.05 &   2 &1.0&   41.2 &     25.2  & 22.9   &   30.8 &  34.7\\
   0.05 &   2 &1.5&   74.0 &     53.6  & 48.4    &  61.2 & 66.4\\
   0.05 &   2 &2.0 &  93.8 &     82.9  & 78.7    &  88.7 & 91.6\\
   0.07 &    1&1.0&   31.2  &    33.3  & 33.1    &  14.8    &15.5\\
   0.07 &    1& 1.5&   66.6 &    61.6  & 61.8    &   55.4   &58.6\\
   0.07 &    1& 2.0&   94.0 &    88.1  & 88.8    &   93.2   &94.7\\
  0.07 &    2 &1.0&   59.3  &    41.6   &38.8    &  49.2    &54.4\\
  0.07 &    2 &1.5&   89.1  &    76.5   &73.3    &  83.4    &86.6\\
  0.07 &    2 &2.0&   99.1  &    96.2  & 95.1    &  97.8    &98.6\\
  0.10 &    1 &1.0 &  54.8  &    57.4  & 57.4    &   26.5   &27.5\\
  0.10 &    1 &1.5 &  89.5  &    88.0  & 88.3    &   82.3   &84.3\\
  0.10 &    1& 2.0 &  99.6  &    98.9  & 99.0    &  99.5    &99.8\\
  0.10 &    2 &1.0 &  79.2  &    64.7  & 62.6    &   74.3   &78.4\\
  0.10 &    2 &1.5 &  98.2  &    94.7  & 93.8    &   96.8   &97.8\\
  0.10 &    2 &2.0 &  99.9  &    99.8  & 99.8    &   99.7   &99.9\\
\hline
\end{tabular} }
\label{chap4 con altertable}
\end{table}
}

\section{A Real-data Example}
\label{chap4 con eg}

The example concerns the police data taken from \cite{efron2010large}. The police data can be downloaded from the website at \url{http://statweb.stanford.edu/~ckirby/brad/LSI/datasets-and-programs/datasets.html}.
%Besides microarray data analysis,
%there are other examples of large-scale inference.
In 2006 at New York City,
a study was conducted to investigate
whether there are some police officers
that have racial bias with pedestrian stops.
The preliminary data included $\bx_{ij}$,
the vector of covariates for police officer $i$, stop $j$;
$y_{ij}$= 0 or 1, the indicator of
whether the stopped person belonged to
a certain minority group or not.
A logistic regression model
$$
\log\frac{\mbox{Pr}(y_{ij}=1)}{1-\mbox{Pr}(y_{ij}=1)}
=\beta_i+\bgamma^{\tau}\bx_{ij}
$$
was used to estimate the ``officer effect" $\beta_i$ \cite{efron2010large}.
The $z$-score of the $i$-th officer is defined as
\[
z_i=\hat \beta_i / \mbox{se}(\hat \beta_i ).
\]
In total, $n=2749$ $z$-scores are obtained.
Large positive $z_i$'s are considered as
signs of possible racial bias.

\begin{figure}[!ht] %  figure placement: here, top, bottom, or page
\caption{Histogram and two fitted densities of the police data:
the density from the homogeneous normal distribution (solid line) and the density from
the two-component contaminated normal distribution (dashed line). }
\centering
\includegraphics[scale=0.5,angle=-90]{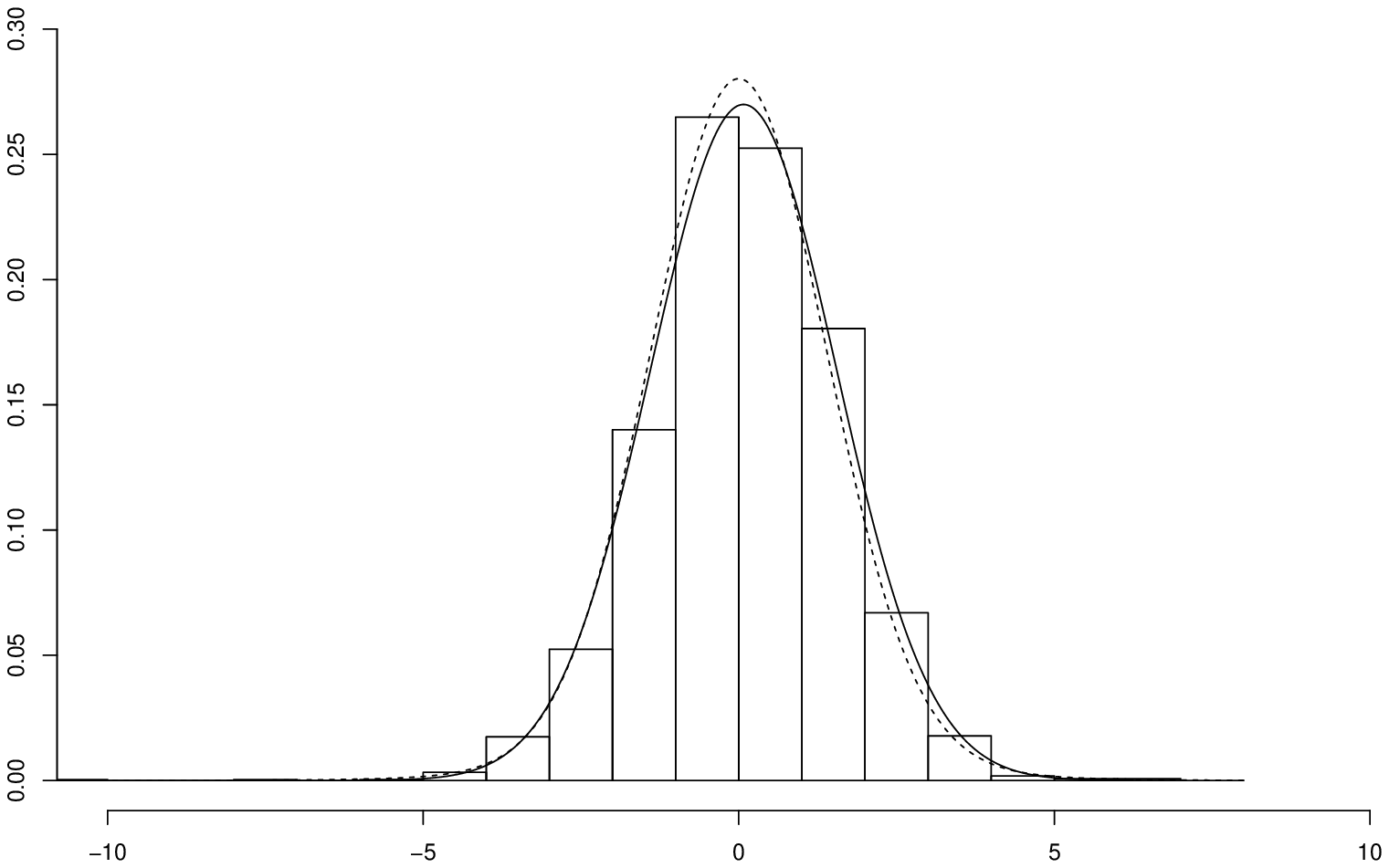}
\label{chap4 con policehist}
\end{figure}

Figure \ref{chap4 con policehist} shows the histogram of
the 2749 $z$-scores.
The solid line and the dashed line are the fitted density curves
for the homogeneous normal distribution and the contaminated normal distribution in (\ref{chap4.model}), respectively.
Merely based on this figure,
it is hard to tell which model provides a better fitting.
Hence a formal test is required here.
The EM-test statistics in this example
are found to be
$EM_n^{(3)}$ =41.042 with the $p$-value being around $1.7e-10$
calibrated by its  limiting distribution.
Based on the $p$-value,
the null hypothesis is soundly rejected.
We also apply the MLRTs and D-tests under both
the contaminated normal model in (\ref{chap4.model2}) and
the contaminated Beta model in (\ref{chap4.model3}) to the police data example.
The $p$-values for the test statistics
$\lambda_{n,N}$, $d_{n,N}$, $\lambda_{n,B}$, and $d_{n,B}$ calibrated by their respect limiting distributions are
respectively $7.2e-06$, $3.9e-1$, $0$, and $0$.
Clearly,
the proposed EM-test provides stronger evidence than $\lambda_{n,N}$ and $d_{n,N}$.
It looks that  $\lambda_{n,B}$  and $d_{n,B}$ provides even stronger evidence.

The fitted contaminated normal model for the 2749 $z$-scores
is
\[
0.951 N(0, 1.391^2)+ 0.049 N(0.021, 2.610^2).
\]
The two component variances are quite different, which explains why the proposed EM-test is more powerful than $\lambda_{n,N}$ and $d_{n,N}$.
Further note that first component distribution is quite different from the theoretical null distribution $N(0,1)$.
Efron (2010) derived the empirical null and found its variance is $1.40^2$, which is quite close to $1.391^2$ but far away from 1.
Both suggest that the  theoretical null distribution may not work here. Table \ref{chap4 con nulltable} in the simulation studies showed that
%To see the effect of the failure of theoretical null distribution on $\lambda_{n,B}$  and $d_{n,B}$,
%we get 10000 random samples of sample size $n=2749$ from $N(0,1.391^2)$
%and find that the simulated type-I error rates of $\lambda_{n,B}$  and $d_{n,B}$ at the 5\% level
%are around 100\%.
in such situations, the limiting distributions of $\lambda_{n,B}$  and $d_{n,B}$ do not provide reasonable approximations and therefore, the results based on $\lambda_{n,B}$  and $d_{n,B}$ are questionable.
However, the EM-test, $\lambda_{n,N}$, and $d_{n,N}$ are invariant to the scale transformation and hence the conclusion based on these three tests are more trustworthy.

\section*{Disclosure statement}

No potential conflict of interest was reported by the authors.

\section*{Funding}

Dr Fu's research was supported by NSERC Discovery Grant RGPIN 2018 05846. Dr Li's
research was supported by NSERC Discovery Grant RGPIN 2015 06592.
%The authors would like to thank the editor, an AE and two referees for their valuable suggestions and comments.
%\par

%%% references
\nocite{*}
\bibliographystyle{abbrvnat}
\bibliography{EMtest}

\begin{thebibliography}{16}
\providecommand{\natexlab}[1]{#1}
\providecommand{\url}[1]{\texttt{#1}}
\expandafter\ifx\csname urlstyle\endcsname\relax
  \providecommand{\doi}[1]{doi: #1}\else
  \providecommand{\doi}{doi: \begingroup \urlstyle{rm}\Url}\fi

\bibitem[Allison et~al.(2002)Allison, Gadbury, Heo, Fern{\'a}ndez, Lee, Prolla,
  and Weindruch]{allison2002mixture}
D.~B. Allison, G.~L. Gadbury, M.~Heo, J.~R. Fern{\'a}ndez, C.-K. Lee, T.~A.
  Prolla, and R.~Weindruch.
\newblock A mixture model approach for the analysis of microarray gene
  expression data.
\newblock \emph{Computational Statistics \& Data Analysis}, 39\penalty0
  (1):\penalty0 1--20, 2002.

\bibitem[Benjamini and Hochberg(1995)]{benjamini1995controlling}
Y.~Benjamini and Y.~Hochberg.
\newblock Controlling the false discovery rate: a practical and powerful
  approach to multiple testing.
\newblock \emph{Journal of the royal statistical society. Series B
  (Methodological)}, pages 289--300, 1995.

\bibitem[Chen and Li(2009)]{chen2009hypothesis}
J.~Chen and P.~Li.
\newblock Hypothesis test for normal mixture models: The em approach.
\newblock \emph{The Annals of Statistics}, 37\penalty0 (5A):\penalty0
  2523--2542, 2009.

\bibitem[Chen and Li(2011)]{chen2011tuning}
J.~Chen and P.~Li.
\newblock Tuning the em-test for finite mixture models.
\newblock \emph{Canadian Journal of Statistics}, 39\penalty0 (3):\penalty0
  389--404, 2011.

\bibitem[Chen et~al.(2008)Chen, Tan, and Zhang]{chen2008inference}
J.~Chen, X.~Tan, and R.~Zhang.
\newblock Inference for normal mixtures in mean and variance.
\newblock \emph{Statistica Sinica}, pages 443--465, 2008.

\bibitem[Dacunha-Castelle and Gassiat(1999)]{dacunha1999testing}
D.~Dacunha-Castelle and E.~Gassiat.
\newblock Testing the order of a model using locally conic parametrization:
  population mixtures and stationary arma processes.
\newblock \emph{The Annals of Statistics}, 27\penalty0 (4):\penalty0
  1178--1209, 1999.

\bibitem[Dai and Charnigo(2008)]{dai2008omnibus}
H.~Dai and R.~Charnigo.
\newblock Omnibus testing and gene filtration in microarray data analysis.
\newblock \emph{Journal of Applied Statistics}, 35\penalty0 (1):\penalty0
  31--47, 2008.

\bibitem[Dai and Charnigo(2010)]{dai2010contaminated}
H.~Dai and R.~Charnigo.
\newblock Contaminated normal modeling with application to microarray data
  analysis.
\newblock \emph{Canadian Journal of Statistics}, 38\penalty0 (3):\penalty0
  315--332, 2010.

\bibitem[Dempster et~al.(1977)Dempster, Laird, and Rubin]{dempster1977maximum}
A.~P. Dempster, N.~M. Laird, and D.~B. Rubin.
\newblock Maximum likelihood from incomplete data via the em algorithm.
\newblock \emph{Journal of the royal statistical society. Series B
  (methodological)}, pages 1--38, 1977.

\bibitem[Efron(2004)]{efron2004large}
B.~Efron.
\newblock Large-scale simultaneous hypothesis testing: the choice of a null
  hypothesis.
\newblock \emph{Journal of the American Statistical Association}, 99\penalty0
  (465):\penalty0 96--104, 2004.

\bibitem[Efron(2010)]{efron2010large}
B.~Efron.
\newblock \emph{Large-scale inference: empirical Bayes methods for estimation,
  testing, and prediction}, volume~1.
\newblock Cambridge University Press, 2010.

\bibitem[Fu et~al.(2008)Fu, Chen, and Li]{fu2008modified}
Y.~Fu, J.~Chen, and P.~Li.
\newblock Modified likelihood ratio test for homogeneity in a mixture of von
  mises distributions.
\newblock \emph{Journal of Statistical Planning and Inference}, 138\penalty0
  (3):\penalty0 667--681, 2008.

\bibitem[Li et~al.(2009)Li, Chen, and Marriott]{li2009non}
P.~Li, J.~Chen, and P.~Marriott.
\newblock Non-finite fisher information and homogeneity: an em approach.
\newblock \emph{Biometrika}, 96\penalty0 (2):\penalty0 411--426, 2009.

\bibitem[Liu and Shao(2003)]{liu2003asymptotics}
X.~Liu and Y.~Shao.
\newblock Asymptotics for likelihood ratio tests under loss of identifiability.
\newblock \emph{The Annals of Statistics}, 31\penalty0 (3):\penalty0 807--832,
  2003.

\bibitem[McLachlan et~al.(2006)McLachlan, Bean, and Jones]{mclachlan2006simple}
G.~J. McLachlan, R.~Bean, and L.~B.-T. Jones.
\newblock A simple implementation of a normal mixture approach to differential
  gene expression in multiclass microarrays.
\newblock \emph{Bioinformatics}, 22\penalty0 (13):\penalty0 1608--1615, 2006.

\bibitem[Peng(2003)]{peng2003simultaneous}
X.~Peng.
\newblock \emph{Simultaneous inference and sample size considerations in
  microarray data analysis.}
\newblock PhD thesis, University of Kentucky, 2003.

\end{thebibliography}

\appendix

\section{Proofs}
\label{chap4 con proof}

Since the  EM-test statistic $EM_n^{(K)}$ is invariant to the scale transformation,
without loss of generality, we assume that under the null hypothesis,
the true distribution is $N(0,1)$.
All the derivations are under this distribution.

We first prove two useful technical lemmas.
Lemma \ref{chap4 con lem1} shows
the consistency of $(\alpha_j^{(K)}, \mu_j^{(K)}, \sigma_{j,1}^{(K)}, \sigma_{j,2}^{(K)})$
and Lemma   \ref{chap4 con lem2} derives an upper bound for
the modified log-likelihood difference, which will be used to derive an upper bound for $EM_n^{(K)}$.
Next, we show the upper bound of  $EM_n^{(K)}$ is achievable
and  derive the limiting distribution of
$EM_n^{(K)}$.

\begin{lemma}
\label{chap4 con lem1}
Suppose Conditions D1-D4 are satisfied.
Then under the null distribution $N(0, 1)$,
we have,
for $j=1, 2, \dots, J$ and any $k \leq K$,
\[
\alpha_j^{(k)}-\alpha_j = o_p(1),
~
\mu_j^{(k)} = o_p(1),
~
\sigma_{j,1}^{(k)} -1 = o_p(1)
~
\textrm{and}
~\sigma_{j,2}^{(k)}-1 = o_p(1).
\]
\end{lemma}

\begin{proof}
The proof is similar to
that of Lemma 6, Lemma 7, and Theorem 3
in Chen and Li (2009).
Thus it is omitted.
\end{proof}

The next lemma concerns the upper bound of the modified log-likelihood difference when
($\mu,\sigma_1,\sigma_2$) are in small neighbourhood  of the true values.
For $i=1, 2, \ldots, n$, we define
\[
Z_i=\frac{X_i^2 - 1}{2},
~
U_i=\frac{X_i^3 - 3X_i}{6},
~
\textrm{and}
~
V_i=\frac{X_i^4 - 6X_i^2 + 3}{24}.
\]

\begin{lemma}
\label{chap4 con lem2}
Assume that the same conditions in
Lemma \ref{chap4 con lem1} hold.
Suppose $(\bar\alpha$, $\bar \mu$, $\bar\sigma_1$, $\bar\sigma_2)$
are estimators of $(\alpha,\mu,\sigma_1,\sigma_2)$
such that
$
(\bar \mu,\bar\sigma_1,\bar\sigma_2)=(0,1,1)+o_{p}(1)$
and
$
\bar\alpha\in (\delta,1-\delta)
$
for some $\delta>0$.
Under the null distribution $N(0,1)$, we have
\begin{eqnarray*}
&&2\{ pl_{n}(\bar\alpha, \bar \mu, \bar\sigma_1, \bar\sigma_2)
- pl_{n}(1, 0,1,1)\} \\
&\leq&
\frac{ (\sum_{i=1}^n X_i )^2}{\sum_{i=1}^n X_i^2}
+
\frac{ (\sum_{i=1}^n Z_i )^2}{\sum_{i=1}^n Z_i^2}
+\frac{ \{(\sum_{i=1}^n V_i )^+\}^2}{\sum_{i=1}^n V_i^2}+2p(\bar\alpha)+o_p(1).
\end{eqnarray*}
\end{lemma}

\begin{proof}
Let
$$
r_{1n}(\alpha,\mu, \sigma_1,\sigma_2 )= 2\{l_n(\alpha,\mu, \sigma_1,\sigma_2 )
-l_n(1, 0, 1, 1)\}.
$$
Then
\begin{eqnarray}
\nonumber
&&2\{ pl_{n}(\bar\alpha, \bar \mu, \bar\sigma_1, \bar\sigma_2)
- pl_{n}(1, 0,1,1)\}\\
&=&r_{1n}(\bar \alpha, \bar \mu, \bar \sigma_1, \bar \sigma_2)+2\{p_n(\bar\sigma_1)+p_n(\bar\sigma_2)+p(\bar\alpha)-2p_n(1)-p(1)\}.
\label{chap4.pln.upper1}
\end{eqnarray}
The upper bounds for the two terms in the above summation will be assessed separately.

We first consider $r_{1n}(\bar \alpha, \bar \mu, \bar \sigma_1, \bar \sigma_2)$.
From (A.20) in Chen and Li (2008), we directly have
\begin{eqnarray*}
r_{1n}(\bar \alpha, \bar \mu, \bar \sigma_1, \bar \sigma_2)&\leq&
2 \Big\{ \bar t_1 \sum_{i=1}^n X_i
+ \bar t_2 \sum_{i=1}^n Z_i
+ \bar t_3 \sum_{i=1}^n U_i
+ \bar t_4 \sum_{i=1}^n V_i \Big\}
\\
&&
- \Big\{ \bar t_1^2 \sum_{i=1}^n X_i^2
+ \bar t_2^2 \sum_{i=1}^n Z_i^2
+ \bar t_3^2 \sum_{i=1}^n U_i^2
+ \bar t_4^2 \sum_{i=1}^n V_i^2 \Big\} \{1+o_p(1)\}\\
&&+o_p(1).
\end{eqnarray*}
In this inequality, $\bar t_l$
are defined by
\[
\bar t_1= \bar m_{1,0},
~
\bar t_2= \bar m_{2,0} +\bar m_{0,1},
~
\bar t_3= \bar m_{3,0} +3\bar m_{1,1},
~
\textrm{and}
~
\bar t_4= \bar m_{4,0} +6\bar m_{2,1} +3\bar m_{0,2},
\]
where $\bar m_{l,s}$ are the first four moments of
the mixing distribution such that
\[
\bar m_{l,s} =
(1-\bar \alpha) 0^l (\bar \sigma_1^2-1)^s
+ \bar \alpha \bar \mu^l (\bar \sigma_2^2-1)^s.
\]
After some simple algebra calculations, we have the following simpler forms of $\bar t_l$:
\begin{eqnarray*}
\bar t_1=\bar \alpha \bar \mu,~
\bar t_2 = \tilde t_2+o_p(\bar t_1),
~\bar t_3 =o_p(\bar t_1),~
\bar t_4 &=&
\tilde t_4+o_p(\bar t_1),
\end{eqnarray*}
where
$\tilde t_2=(1-\bar  \alpha)(\bar  \sigma_1^2-1)
+ \bar  \alpha(\bar  \sigma_2^2-1) $
and
$\tilde t_4=3\{(1-\bar  \alpha)(\bar  \sigma_1^2-1)^2
+ \bar  \alpha(\bar  \sigma_2^2-1)^2\}$.
Hence
\begin{eqnarray}
\nonumber r_{1n}(\bar \alpha, \bar \mu, \bar \sigma_1, \bar \sigma_2)&\leq&
2 \Big\{ \bar t_1 \sum_{i=1}^n X_i
+ \tilde t_2 \sum_{i=1}^n Z_i
+ \tilde t_4 \sum_{i=1}^n V_i \Big\}
\\
&&
\nonumber
- \Big\{ \bar t_1^2 \sum_{i=1}^n X_i^2
+ \tilde t_2^2 \sum_{i=1}^n Z_i^2
+ \tilde t_4^2 \sum_{i=1}^n V_i^2 \Big\} \{1+o_p(1)\}\\
&&+o_p(1).
\label{chap4.r1n1}
\end{eqnarray}

We now assess the upper bound for $2\{p_n(\bar\sigma_1)+p_n(\bar\sigma_2)+p(\bar\alpha)-2p_n(1)-p(1)\}$.
By Conditions D1 and D3,
we have
\begin{eqnarray}
\nonumber
&&2\{p_n(\bar\sigma_1)+p_n(\bar\sigma_2)+p(\bar\alpha)-2p_n(1)-p(1)\}\\
\nonumber&\leq& 2\{p_n(\bar\sigma_1)+p_n(\bar\sigma_2)-2p_n(1)+p(\bar\alpha)\}=o_{p}(n^{1/6})\{|\bar\sigma^2_1-1|+|\bar\sigma^2_2-1|\}+2p(\bar\alpha)\\
&\leq& 2p(\bar\alpha)+
o_p(n)\Big\{ \tilde t_2^2+\tilde t_4^2 \Big\}
+o_p(1).
\label{chap4.r3n}
\end{eqnarray}

Combining (\ref{chap4.pln.upper1})-(\ref{chap4.r3n}),
we get
\begin{eqnarray}
\nonumber
&&2\{ pl_{n}(\bar\alpha, \bar \mu, \bar\sigma_1, \bar\sigma_2) -pl_{n}(1,0,1,1)\}\\
&\leq&2 \Big\{ \bar t_1 \sum_{i=1}^n X_i
+ \tilde t_2 \sum_{i=1}^n Z_i
+ \tilde t_4 \sum_{i=1}^n V_i \Big\} \nonumber
\\
&&
- \Big\{ \bar t_1^2 \sum_{i=1}^n X_i^2
+ \tilde t_2^2 \sum_{i=1}^n Z_i^2
+ \tilde t_4^2 \sum_{i=1}^n V_i^2 \Big\} \{1+o_p(1)\}+2p(\bar\alpha)+o_p(1).
\label{chap4.pln.upper2}
\end{eqnarray}

Define
\begin{eqnarray*}
&&Q(t_1, t_2, t_4) \\
&=&2\Big\{ t_1 \sum_{i=1}^n X_i
+ t_2 \sum_{i=1}^n Z_i
+ t_4 \sum_{i=1}^n V_i \Big\}
- \Big\{ t_1^2 \sum_{i=1}^n X_i^2
+ t_2^2 \sum_{i=1}^n Z_i^2
+ t_4^2 \sum_{i=1}^n V_i^2 \Big\}
\end{eqnarray*}
as a function of $(t_1, t_2, t_4)$ with $t_4\geq0$.
With
\begin{equation}
\label{chap4.hatt}
\hat t_1= \frac{\sum_{i=1}^n X_i}{\sum_{i=1}^n X_i^2},
~
\hat t_2 = \frac{\sum_{i=1}^n Z_i}{\sum_{i=1}^n Z_i^2},
~
\textrm{and}
~
\hat t_4 = \frac{(\sum_{i=1}^n V_i )^+}{\sum_{i=1}^n V_i^2},
\end{equation}
this quadratic function is
maximized.
Further the maximized value of $Q(t_1, t_2, t_4)$ is
\[
Q(\hat t_1, \hat t_2, \hat t_4)
= \frac{ (\sum_{i=1}^n X_i)^2}{\sum_{i=1}^n X_i^2}
+ \frac{ (\sum_{i=1}^n Z_i)^2}{\sum_{i=1}^n Z_i^2}
+ \frac{ \{ (\sum_{i=1}^n V_i)^+ \}^2}{\sum_{i=1}^n V_i^2}.
\]
Together with (\ref{chap4.pln.upper2}), we get
\begin{eqnarray}
\nonumber
&&2\{ pl_{n}(\bar\alpha, \bar \mu, \bar\sigma_1, \bar\sigma_2) -pl_{n}(1,0,1,1)\}\\
&\leq& Q(\hat t_1, \hat t_2, \hat t_4)+2p(\bar\alpha)+o_p(1)\nonumber\\
&\leq&\frac{ (\sum_{i=1}^n X_i)^2}{\sum_{i=1}^n X_i^2}
+ \frac{ (\sum_{i=1}^n Z_i)^2}{\sum_{i=1}^n Z_i^2}
+ \frac{ \{ (\sum_{i=1}^n V_i)^+ \}^2}{\sum_{i=1}^n V_i^2}+2p(\bar\alpha)+o_p(1). ~~
\nonumber
\end{eqnarray}
This finishes the proof.
\end{proof}

We now move to the proof of Theorem \ref{chap4 con thm2}.
The consistency results in Lemma \ref{chap4 con lem1}
enable us to apply Lemma \ref{chap4 con lem2}
to
$
2\{pl_{n}(\alpha_j^{(K)}, \mu_j^{(K)}, \sigma_{j,1}^{(K)}, \sigma_{j,2}^{(K)}) -pl_{n}(1,0,1,1) \}.
$
That is,
\begin{eqnarray}
\nonumber
&&2\{ pl_{n}(\alpha_j^{(K)}, \mu_j^{(K)}, \sigma_{j,1}^{(K)}, \sigma_{j,2}^{(K)}) -pl_{n}(1,0,1,1)\}\\
&\leq&\frac{ (\sum_{i=1}^n X_i)^2}{\sum_{i=1}^n X_i^2}
+ \frac{ (\sum_{i=1}^n Z_i)^2}{\sum_{i=1}^n Z_i^2}
+ \frac{ \{ (\sum_{i=1}^n V_i)^+ \}^2}{\sum_{i=1}^n V_i^2}+2p(\alpha_j)+o_p(1). ~~
\nonumber
\end{eqnarray}
Note that classic theory for regular models implies
\begin{equation*}
\label{chap4 con r2n}
2\{ pl_{n}(1,  0, \hat\sigma_{0},\hat\sigma_{0}) -pl_{n}(1,0,1,1)\} =
\frac{ (\sum_{i=1}^n Z_i)^2}{\sum_{i=1}^n Z_i^2}
+ o_p(1).
\end{equation*}
Hence
\begin{eqnarray}
\nonumber
M_n^{(K)}(\alpha_j)&=&2\{ pl_{n}(\alpha_j^{(K)}, \mu_j^{(K)}, \sigma_{j,1}^{(K)}, \sigma_{j,2}^{(K)}) -pl_{n}(1,  0, \hat\sigma_{0},\hat\sigma_{0}) \}\\
\nonumber&=&2\{ pl_{n}(\alpha_j^{(K)}, \mu_j^{(K)}, \sigma_{j,1}^{(K)}, \sigma_{j,2}^{(K)}) -pl_{n}(1,0,1,1)\}\\
\nonumber&&-2\{ pl_{n}(1,  0, \hat\sigma_{0},\hat\sigma_{0}) -pl_{n}(1,0,1,1)\}\\
&\leq&\frac{ (\sum_{i=1}^n X_i)^2}{\sum_{i=1}^n X_i^2}
+ \frac{ \{ (\sum_{i=1}^n V_i)^+ \}^2}{\sum_{i=1}^n V_i^2}+2p(\alpha_j)+o_p(1). ~~
\nonumber
\end{eqnarray}
The upper bound of $EM_n^{(K)}$ is then given by
\begin{eqnarray}
EM_n^{(K)}&\leq&\frac{ (\sum_{i=1}^n X_i)^2}{\sum_{i=1}^n X_i^2}
+ \frac{ \{ (\sum_{i=1}^n V_i)^+ \}^2}{\sum_{i=1}^n V_i^2}+2\max_jp(\alpha_j)+o_p(1). ~~
\label{chap4.em}
\end{eqnarray}

Next, we show that the upper bound in (\ref{chap4.em})
is achievable.
Since  the EM-iteration increases the modified likelihood (Dempster, Laird, and Rubin (1977)),
we need only show that this is the case when $K=1$.
If suffices to find a set of parameter values $\hat\alpha$ and $(\hat\mu,\hat\sigma_1^2,\hat\sigma_2^2)$
 at which the upper bound (\ref{chap4.em}) is
attained.

 We first choose $\hat\alpha$ such that $p(\hat\alpha)=\max_jp(\alpha_j)$.
 Without loss of generality, we assume that $\hat\alpha=\alpha_1$.
 We next choose $(\hat\mu,\hat\sigma_1^2,\hat\sigma_2^2)$ such that
 \begin{displaymath}
\left\{ \begin{array}{l}
\label{chap4 con equgroup}
\hat\alpha \hat \mu=\hat t_1,\\
(1-\hat \alpha)(\hat \sigma_1^2-1)+ \hat \alpha (\hat \sigma_2^2-1)=\hat t_2,\\
3\{ (1-\hat \alpha) (\hat \sigma_1^2-1)^2+ \hat\alpha (\hat \sigma_2^2-1)^2\}=\hat {t_4},
\end{array} \right.
\end{displaymath}
where the expressions of $\hat t_l$ are given  in (\ref{chap4.hatt}).
It can be checked that $(\hat\mu,\hat\sigma_1^2,\hat\sigma_2^2)$ exists
and
$$\hat\mu =O_p(n^{-1/2}),~
\hat \sigma_1^2 -1 =O_p(n^{-1/4}), ~
\hat \sigma_2^2 -1 =O_p(n^{-1/4}). $$
With these order information,
we obtain
\begin{eqnarray*}
&&2\{ pl_{n}(\hat\alpha, \hat\mu,\hat\sigma_1^2,\hat\sigma_2^2)
- pl_{n}(1, 0, \hat \sigma_0, \hat \sigma_0)\} \\
&=&
\frac{ (\sum_{i=1}^n X_i )^2}{\sum_{i=1}^n X_i^2}
+\frac{ \{(\sum_{i=1}^n V_i )^+\}^2}{\sum_{i=1}^n V_i^2}
+2p(\hat\alpha)
+ o_p(1).
\end{eqnarray*}
Note that $\hat\alpha=\alpha_1$ and $p(\hat\alpha)=p(\alpha_1)=\max_jp(\alpha_j)$.
Thus for $EM_n^{(K)}$, we have
\begin{eqnarray*}
EM_n^{(K)}
&\geq& M_n^{(1)}(\alpha_1)
\geq
2\{ \sup_{\mu, \sigma_1, \sigma_2} pl_{n}(\alpha_1, \mu, \sigma_1, \sigma_2) - pl_{n}(1, 0, \hat \sigma_0, \hat \sigma_0) \}
\\
&\geq&
2\{ pl_{n}(\hat\alpha, \hat\mu,\hat\sigma_1^2,\hat\sigma_2^2)
- pl_{n}(1, 0, \hat \sigma_0, \hat \sigma_0)\}
\\
&=&
\frac{ (\sum_{i=1}^n X_i )^2}{\sum_{i=1}^n X_i^2}
+\frac{ \{(\sum_{i=1}^n V_i )^+\}^2}{\sum_{i=1}^n V_i^2}
+2 \max_jp(\alpha_j)
+ o_p(1).
\end{eqnarray*}
This shows the asymptotic upper bound
of $EM_n^{(K)}$ is also the asymptotic lower bound,
which implies
\[
EM_n^{(K)}
=
\frac{ (\sum_{i=1}^n X_i )^2}{\sum_{i=1}^n X_i^2}
+\frac{ \{(\sum_{i=1}^n V_i )^+\}^2}{\sum_{i=1}^n V_i^2}
+2 \max_jp(\alpha_j)
+ o_p(1).
\]
By the central limit theorem,
both
$
{ \sum_{i=1}^n X_i }\Big/{\sqrt{\sum_{i=1}^n X_i^2} }
$
and
$
{ \sum_{i=1}^n V_i }\Big/{\sqrt{\sum_{i=1}^n V_i^2} }
$
converge in distribution to $N(0,1)$.
Further $X_i$ and $V_i$ are uncorrelated.
Therefore, $EM_n^{(K)}$ asymptotically follows
the distribution
\[
\frac{1}{2} \chi_1^2 +\frac{1}{2} \chi_2^2
+2 \max_jp(\alpha_j).
\]
This finishes the proof.
%\qed

\end{document}